\documentclass[aps,prx,twocolumn,superscriptaddress,reprint,amsmath,amssymb]{revtex4-2}

\usepackage{graphicx}
\usepackage{dcolumn}
\usepackage{bm}
\usepackage{amsmath}
\usepackage{amssymb}
\usepackage{latexsym}
\usepackage{epsfig}
\usepackage{amsbsy}
\usepackage{array}
\usepackage{amssymb}
\usepackage{setspace}
\usepackage{mathtools}
\usepackage{amsthm}
\usepackage{physics}
\usepackage{algorithm}
\usepackage{algpseudocode}
\usepackage[colorlinks=true, urlcolor=blue, citecolor=black]{hyperref}

\usepackage{amsthm}
\newtheorem{theorem}{Theorem}

\usepackage{apptools}
\AtAppendix{\counterwithin{lemma}{section}}
\AtAppendix{\counterwithin{theorem}{section}}

\usepackage{algpseudocode}

\usepackage[justification=centerlast]{caption}
\usepackage[labelformat=simple]{subcaption}

\usepackage[justification=centerlast]{caption}

\usepackage{xcolor}
\newcommand{\xiangyi}[1]{\textcolor{blue}{#1}}

\begin{document}
	
	\title{Multipartite Entanglement Routing as a Hypergraph Immersion Problem}

	\author{Yu Tian}%
	\affiliation{Nordita, Stockholm University and KTH Royal Institute of Technology, SE-10691 Stockholm, Sweden}
	\affiliation{Center for Systems Biology Dresden, 01307 Dresden, Germany}
	
	\author{Yuefei Liu}%
	\affiliation{Department of Applied Physics, School of Engineering Sciences, KTH 
		Royal Institute of Technology, AlbaNova University Center, SE-10691 Stockholm, Sweden}
	
	\author{Xiangyi~Meng}%
	\email{xmenggroup@gmail.com}
	\affiliation{Department of Physics, Applied Physics, and Astronomy, Rensselaer Polytechnic Institute, Troy, New York 12180, USA}%
	
	\date{\today}
	
	\begin{abstract}
		Multipartite entanglement, linking multiple nodes simultaneously, is a higher-order correlation that offers advantages over pairwise connections in quantum networks (QNs). Creating reliable, large-scale multipartite entanglement requires entanglement routing, a process that combines local, short-distance connections into a long-distance connection, which can be considered as a transformation of network topology.
		Here, we address the question of whether a QN can be topologically transformed into another via entanglement routing. Our key result is an exact mapping from multipartite entanglement routing to Nash-Williams's graph immersion problem, extended to hypergraphs. This generalized hypergraph immersion problem introduces a partial order between QN topologies, permitting certain topological transformations while precluding others, offering discerning insights into the design and manipulation of higher-order network topologies in QNs.
	\end{abstract}
	
	\maketitle
	
	Quantum entanglement, a pivotal quantum resource~\cite{q-resour_cg19}, 
	allows for the simultaneous correlation of \emph{multiple} qubits. Multipartite entanglement offers notable benefits, in terms of both efficiency and scalability, over bipartite entanglement in a variety of quantum information applications such as sharing secrets among multiple parties~\cite{multipartite-secret-share_fmfp07}
	and reducing the memory requirements for entangling qubits~\cite{miguel2023optimized}. These benefits make multipartite entanglement a necessary building block for future quantum networks (QN)~\cite{multipartite-q-netw_hpe19,multipartite-q-netw_wxkhlhgsp20,*multipartite-q-netw_ctpv21,multipartite-commun-q-netw}.
	In a QN, each \emph{vertex (node)} typically represents a local assembly of qubits entangled with qubits from other vertices, and each \emph{edge (link)} signifies an entangled state between qubits that belong to different vertices~\cite{QEP_acl07,*conpt_mgh21}.
	In this setting, multipartite entanglement is most effectively represented as a higher-order interaction, or in the language of hypergraph theory, as a \emph{hyperedge}~\cite{hypergr-theor-introd}. 
	Contrasting with ordinary edges that connect just two vertices, hyperedges can connect several vertices at once---a feature that has wider applications in not only physics but also chemistry~\cite{jost2019hyperL}, biology~\cite{klamt2009biology}, and social sciences~\cite{taramasco2010collab,*krumov2011collab, *bianconi2021hyper, *Kwang2023core}.

	In QNs, vertices not directly connected by an edge can still be entangled through a fundamental operation known as {entanglement routing}~\cite{q-netw-route_p19,*q-netw-route_pkttjbeg19,*q-netw-route_lphnm20,multipartite-q-netw-route_sb23}, provided that there is a path of edges connecting the vertices.
	This operation, however, utilizes all the edges along the path, subsequently deleting them from the QN, leading to a dynamic process called path percolation~\cite{path-percolation_mhrk24}. 
	As a result, each entanglement routing operation uniquely modifies the QN's topology in an irreversible manner, permitting certain types of topological transformations while precluding others, raising the question: \emph{can a QN's topology be transformed from one to another through a series of entanglement routing operations?} 
	It is also unknown how this question extends to hypergraph QNs 
	due to the ambiguous definition of paths involving multipartite entanglement~\cite{multipartite-q-netw-route_sb23}. 
	Understanding this question is key to grasping the complexity of many-body entanglement structures (e.g.,~matrix product states~\cite{ho2019periodic}, which can be reformulated as QN representations~\cite{multipartite-commun-q-netw, conpt_mhtdlgh23}). Moreover, the insights can be applied to evaluate the effectiveness of qubit layouts (e.g.,~the hex lattices used by IBM~\cite{chow2021ibm}). This is especially relevant as small-scale multipartite-entanglement QNs have been implemented on superconducting~\cite{huber2024parametric, *chen2014qubit} and optical~\cite{su2016quantum, *kok2007linear} platforms.

	Here, we address this question of entanglement routing, involving \emph{both} edges and hyperedges, by establishing an exact mapping to a topological problem on hypergraphs. We first identify and simplify entanglement routing into two hypergraph rules by considering the routing operation for multipartite entanglement. 
	Then, we show that the corresponding topological problem can be viewed as a generalization of the well-established \emph{graph immersion} problem~\cite{RobertsonSeymour_minor_2010}, which, to our best knowledge, has not been extended to hypergraphs. This aptly allows us to formalize our problem as \emph{hypergraph immersion}, which establishes a {partial order}~\cite{algebra-graph-theor} between hypergraph topologies.
	Following graph immersion~\cite{embed-fix-param-tract_gkmw11}, we also prove in an accompanying work that the algorithmic complexity of solving hypergraph immersion is polynomial~\cite{meng2024tractable}. 
	This unique bridge between quantum information and graph theory may provide some useful insights into the efficient design of multipartite QNs.
	
	\begin{figure}[t!]
		\centering
		\begin{minipage}[b]{120pt}
			\begin{minipage}[b]{120pt}
				\centering
				{\includegraphics[width=120pt]{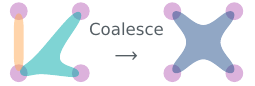}\subcaption{\label{fig_coalesce}}}
			\end{minipage} \\[20pt] 
			
			\begin{minipage}[b]{120pt}
				\centering
				{\includegraphics[width=120pt]{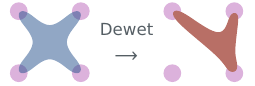}\subcaption{\label{fig_dewet}}}
			\end{minipage}
		\end{minipage}\hspace{2mm}
		\begin{minipage}[b]{114pt}
			\centering
			{\includegraphics[width=114pt]{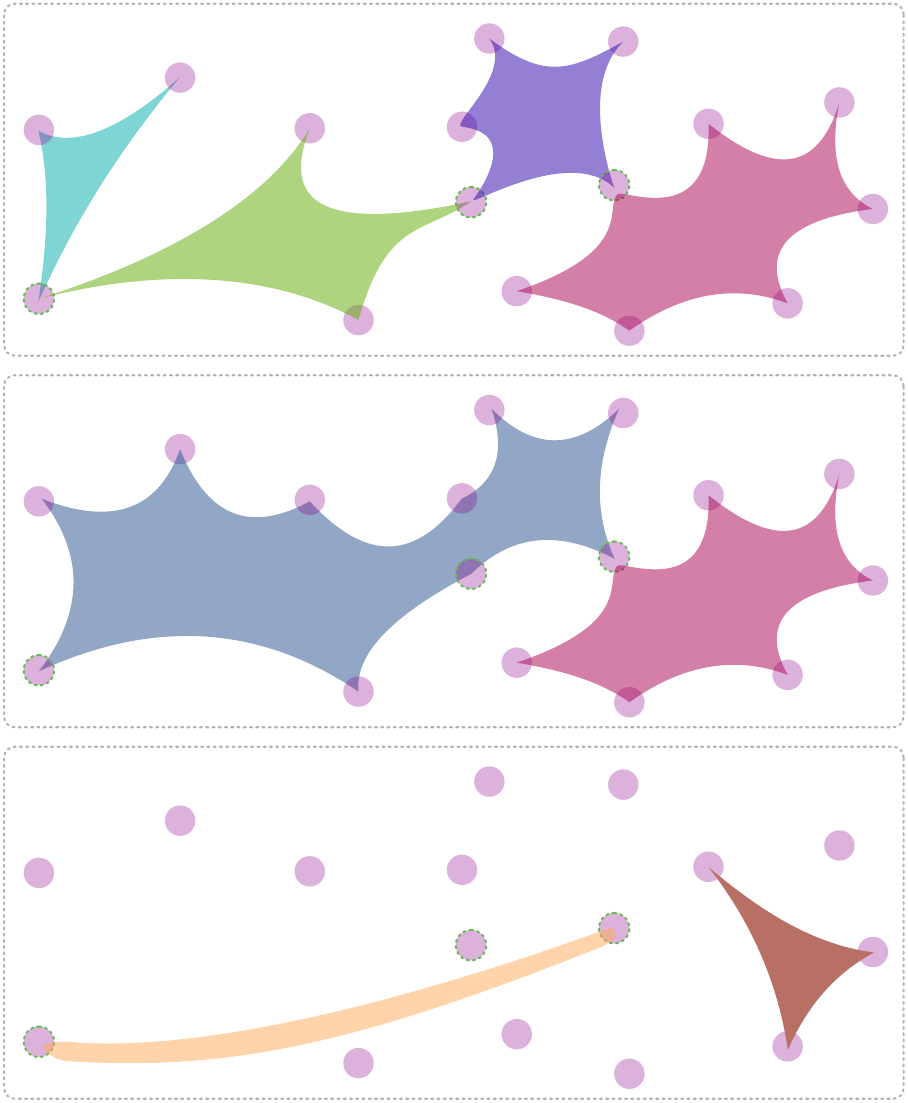}\subcaption{\label{fig_coa_eva}}}
		\end{minipage}
		\vspace{-3mm}
		\caption{\label{fig_operation}\textbf{Hypergraph immersion.} 
			(a) The ``coalescence'' operation merges two  hyperedges that share at least one vertex, resulting in a new hyperedge that joins all vertices from the original two. (b) The ``dewetting'' operation detaches a vertex from a hyperedge, reducing the hyperedge's size by one. (c) 
			The desired final QN topology can be reached from the initial topology through a sequence of operations, provided that the final topology is immersed within the original structure.\vspace{-2mm}
			\hfill\hfill
		}
	\end{figure}

	\emph{Hypergraph QN.---}We start by considering a hypergraph QN model where each hyperedge $e$ of size $r$ (i.e.,~$e$ joins $r$ vertices together) is represented as a generalized GHZ state between $r$ qubits:
	\begin{eqnarray}
		\left|\text{GHZ}\right\rangle_r = \frac{|\stackrel{r}{\overbrace{00\dots0}}\rangle + |\stackrel{r}{\overbrace{11\dots1}}\rangle}{\sqrt{2}}.
	\end{eqnarray}
	These states, albeit varying in size ($r\ge 2$), always have the lowest nontrivial \emph{generalized Schmidt rank}~\cite{nielsen2010quantum, *eisert2001schmidt} to exhibit entanglement  ($\text{rank}=2$). 
	GHZ states are prevalent in quantum systems, such as the random transverse-field quantum Ising model, where the critical ground state consists exclusively of GHZ-state magnetic domains~\cite{sdrg_ck23}. The Ising ground state naturally facilitates a QN interpretation using complex network theory~\cite{QEP-GHZ_pclla10}.
	
	How do we entangle multiple vertices by \emph{routing} via several hyperedges? 
	Given, for example, two GHZ states of $r_1$ and $r_2$ qubits, as well as at least one common vertex that shares two qubits separately from the two states,
	a simple projector on the two qubits,
	\begin{eqnarray}
		M \propto \ket{0}\!\bra{00} + \ket{1}\!\bra{11}, 
		\label{eq:measurement_coa}
	\end{eqnarray}
	suffices to ``merge'' the two states into a larger GHZ state of $r_1+r_2-1$ qubits, which are all simultaneously entangled.
	Conversely, one has the freedom to {disentangle} a vertex from a GHZ state of size $r$ by applying
	\begin{eqnarray}
		M'  \propto \ket{0}\!\bra{0} + \ket{0}\!\bra{1} 
		\label{eq:measurement_eva}
	\end{eqnarray}
	on the vertex, yielding a smaller GHZ state with $r-1$ qubits.
	With appropriate rotations of $\ket{0}$ and $\ket{1}$,
	these two operations are also applicable to general {Schmidt-rank-2} states of the form $| \alpha_1 \otimes \alpha_2 \otimes ... \otimes \alpha_n \rangle + z |\beta_l \otimes \beta_2 \otimes ... \otimes \beta_n \rangle$~\cite{rank-2_tgp10}. An example of such is the cat state, commonly used in continuous-variable optical communications~\cite{deleglise2008reconstruction}.  
	
	The projectors $M$ and $M'$ do {not} represent optimal routing protocols. However,
	a combination of $M$ and $M'$ is sufficient to reproduce, with a consistently nonzero success probability, any \emph{topological} change to QN through entanglement routing tasks. Therefore, entanglement routing, or simply the inclusion of $M$ and $M'$, creates a {partial order} between different QN topologies made up of general Schmidt-rank-2 states. This is akin to, but inherently different from how stochastic local operations and classical communication (SLOCC) sets an \emph{algebraic} partial order for generic multipartite entanglement~\cite{szalay2015multipartite}.
	
	\emph{Coalescence and dewetting operations.---}The topological emphasis of  Eqs.~\eqref{eq:measurement_coa}~and~\eqref{eq:measurement_eva} suggests that we can simplify them as hypergraph operations. Consider two hypergraphs $H$ and $G$, each comprising sets of vertices $V(H)$ and $V(G)$, and sets of hyperedges $E(H)$ and $E(G)$, respectively. Both $H$ and $G$ are loopless~\footnote{We note that a loop in (hyper)graphs means an edge contains the same vertex more than once. The case where an edge only includes one vertex is allowed.} and undirected, but may include multi-hyperedges. We claim that the feasibility of deriving $H$ from $G$ through entanglement routing is equivalent to asserting that $H$ can be derived from $G$ through the following hypergraph operations:
	\begin{itemize}
		\item \emph{Coalescence}: merging two connected hyperedges to create a new one including all the incident vertices.
		\item \emph{Dewetting}: detaching a hyperedge from one vertex that it is incident on.
	\end{itemize}
	
	One can visualize vertices in the hypergraph as spatially distributed objects (pink) that are ``hydrophilic'' (Fig.~\ref{fig_operation}).
	Correspondingly, each hyperedge $e$ can be visualized as a water droplet, touching those hydrophilic objects (vertices) that are incident with $e$. In this physical analogy, the coalescence operation results in the merging of two droplets that share contact with a common vertex $v$, forming a larger droplet [Fig.~\ref{fig_coalesce}]; the dewetting operation simulates the reduction in the volume of a droplet and the loss of its connections to the vertices it was previously in contact with [Fig.~\ref{fig_dewet}]. This physical analogy motivates our naming of the two operations.
	
	Note that both operations are irreversible: the coalescence reduces the number of hyperedges by one; the dewetting reduces the number of vertices incident with the hyperedge by one. 
	Hence, for $H$ to be derived from $G$ through these operations, a necessary condition is $\left|V(H)\right|\le \left|V(G)\right|$ and $\left|E(H)\right|\le \left|E(G)\right|$, i.e.,~$H$ must not possess more vertices or hyperedges than $G$. In general, determining whether $H$ can be derived from $G$ poses a nontrivial problem. In the following, we present an alternative definition in parallel to these operations.
	
	\emph{Hypergraph immersion.---}We  consider a general function, $\alpha$, with domain $V(H)\cup E(H)$, such that: 
	\begin{enumerate}
		\item $\alpha(v)\in V(G)$ for all $v\in V(H)$, and $\alpha(v_1)\neq \alpha(v_2)$ for all distinct $v_1,v_2\in V(H)$; 
		\item for each hyperedge $e\in E(H)$, if $e$ has distinct ends $v_1,v_2,\cdots$, then $\alpha(e)$ is a connected subgraph in $G$ that includes $\alpha(v_1),\alpha(v_2),\cdots$. 
		\item for all distinct $e_1,e_2\in E(H)$, $E(\alpha(e_1)\cap \alpha(e_2))=\emptyset$;
	\end{enumerate} 
	In other words, $\alpha$ is an injective mapping from vertices in $H$ to vertices in $G$, and from hyperedges in $H$ to \emph{edge-disjoint} connected subgraphs in $G$.
	The formulation of $\alpha$ mirrors the \emph{graph immersion} as outlined in Ref.~\cite{RobertsonSeymour_minor_2010}, which has a similar definition that maps edges in $H$ to edge-disjoint paths in $G$, but with $H$ and $G$ being ordinary graphs. For ordinary edges, the function $\alpha$ reduces to the definition of graph immersion. More (hyper)graph theory terminologies can be found in SM, Section~1.
	\begin{theorem}
		An immersion $\alpha$ of hypergraph $H$ in hypergraph $G$ exists if and only if $H$ can be obtained from a subgraph of $G$ by a sequence of coalescence and dewetting operations. 
	\end{theorem}
	\begin{proof}
		If an immersion $\alpha$ exists, we can show that $H$ is isomorphic to a hypergraph $G^{\text{cl;dw}}$ obtained by applying the coalescence and dewetting operations on a subgraph of $G$, where the bijection corresponding to the isomorphism $f: V(H)\to V(H')$ takes the same value as $\alpha$, i.e.,~$\forall v\in V(H)$, $f(v) = \alpha(v)$. 
		
		We build $G^{\text{cl;dw}}$ as follows. 
		For each $e\in E(H)$ with distinct ends $v_1, v_2, \dots, v_{\left|e\right|}$, we apply the coalescence operation to every adjacent pair of hyperedges in the corresponding connected component $\alpha(e)\subseteq G$, until we obtain one single hyperedge incident to all vertices in $\{\alpha(v_1), \alpha(v_2), \dots, \alpha(v_{\left|e\right|})\}$. The coalescence operation is possible since $\alpha(e)$ is connected.
		Then, we apply the dewetting operation to detach the single hyperedge from vertex $u$, $\forall u\in V(\alpha(e))\backslash \{\alpha(v_1), \alpha(v_2), \dots, \alpha(v_{\left|e\right|})\}$. Doing this for each edge-disjoint $\alpha(e)$, $\forall e\in E(H)$, we denote the resulted hypergraph, a collection of single hyperedges derived from every $\alpha(e)$, by $G^{\text{cl;dw}}$.
		It is straightforward to show that $H$ is isomorphic to $G^{\text{cl;dw}}$ with the bijection $f$. 
		
		Now, if $H$ can be obtained by applying the two operations to a subgraph of $G$, we can immediately construct $\alpha$ which maps $v\in V(H)$ to $\alpha(v)\in V(G)$, the one it comes from. Then $a(v_1)\ne \alpha(v_2)$ for all distinct $v_1,v_2\in V(H)$. Also, for each edge $e\in E(H)$, we can retrieve the set of edges $\alpha(e)\in E(G)$ where we apply the operations to obtain $e$. Then $\alpha(e)$ is connected since the two operations can only be applied to connected edges. Also, for all distinct $e_1, e_2\in E(H)$, $E(\alpha(e_1)\cap\alpha(e_2)) = \emptyset$, since each edge in $G$ can be used at most once in the operations. Hence, $H$ is immersed in $G$ by definition. 
	\end{proof}
	
	Intuitively, the three conditions in the definition of hypergraph immersion can be understood in parallel with concepts in quantum information:
	
	The first condition, concerning vertices---or \emph{locality}, requires that
	two vertices in $H$ are distinct if and only if they correspond to different vertices in $G$. In other words, $H$ cannot introduce or merge vertices. This constraint underscores the fact that only qubits associated with different QN vertices are restricted by (S)LOCC---the free operation that establishes entanglement as a resource~\cite{q-resour_cg19}. 
	This {locality} constraint in QN must be upheld throughout entanglement routing.
	
	The second condition, concerning hyperedges---or \emph{connectivity}, requires that to establish entanglement between distinct vertices, the vertices must be ``connected.'' 
	This emphasis loosely resonates with the Reeh--Schlieder theorem~\cite{reeh-schlieder_rs61}, suggesting that localized regions in quantum field theory (QFT) must also be  ``connected''~\cite{q-field-theor-entangle_w18}.  Indeed, we can view a QN as a partition of spatially distributed qubits (quantum fields), such that qubits are affiliated with the same vertex in QN if and only if they are encompassed within the same local region.
	This view translates QN to the language of QFT~\cite{q-field-theor-entangle_cfhrw98}, emphasizing the topological aspects of entanglement overall.
	
	The third condition requires that 
	each hyperedge in $G$ can be utilized only once to derive a hyperedge in $H$. This rule is a direct reflection of the \emph{no-cloning theorem}~\cite{q-no-clone_p70,*q-no-clone_wz82}, asserting that a quantum state cannot be replicated. In other words, once a hyperedge is utilized for entanglement routing, it cannot be reused.
	
	\begin{table}[t]
		\centering
		\fontsize{9}{10}\selectfont
		\caption{\label{table_k}
			\textbf{Conditions for immersion of a complete $r$-uniform hypergraph $K_n^r$ in $G$.\hfill\hfill} 
		}
		\begin{tabular}{p{22pt}|c|p{200pt}}
			\hline\hline
			\multicolumn{2}{c|}{$H=K_n^r$ }
			& $H$ can be immersed in $G$ if and only if...\\
			\hline
			\includegraphics[width=22pt]{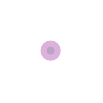}&$K_1^1$&$\left|V(G)
			\right|\ge1$.\\
			\hline
			\includegraphics[width=22pt]{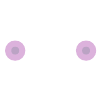}&$K_2^1$&$\left|V(G)
			\right|\ge2$.\\
			\hline
			\includegraphics[width=22pt]{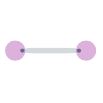}&$K_2^2$&$\left|E(G)
			\right|\ge 1$.\\
			\hline
			\includegraphics[width=22pt]{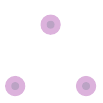}&$K_3^1$&$\left|V(G)
			\right|\ge3$.\\
			\hline
			\includegraphics[width=22pt]{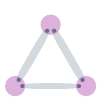}&$K_3^2$&\vspace{-6mm}$\exists$ a Berge cycle~\cite{berge1973hyper} of length $\ge3$, or two length-$2$ Berge cycles that are edge-disjoint and share exactly one common vertex.
			\\
			\hline
			\includegraphics[width=22pt]{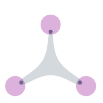}&$K_3^3$&$\exists$ a connected subgraph of size $\ge3$ in $G$.\\
			\hline
			\includegraphics[width=22pt]{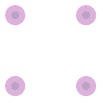}&$K_4^1$&$\left|V(G)
			\right|\ge4$.\\
			\hline
			\includegraphics[width=22pt]{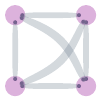}&$K_4^2$&\vspace{-6mm}{$\exists$ an (ordinary) graph, derived through dewetting only on hyperedges in $G$, that is either non-series-parallel~\cite{Duffin_seriesparallel_1965} or falls into the special cases given  in Ref.~\cite{Booth_K4_1999} (SM, Section~2.1).}\\
			\hline
			\includegraphics[width=22pt]{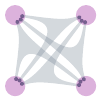}&$K_4^3$&\vspace{-6mm}{$\exists$ a ``restricted'' graph immersion of one of $18$ ordinary graph variants of $K_4^3$ (SM, Section~2.2).}\\
			\hline
			\includegraphics[width=22pt]{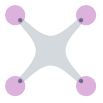}&$K_4^4$&$\exists$ a connected subgraph of size $\ge4$ in $G$.\\
			\hline
			$\vdots$&$\vdots$&$\vdots$\\
			\hline\hline
		\end{tabular}
	\end{table}
	
	\begin{figure*}[t!]
		\centering
		\begin{minipage}[b]{121pt}
			\centering
			{\includegraphics[width=121pt]{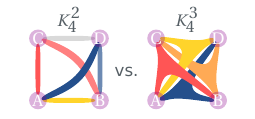}\vspace{-3mm}\subcaption{\label{fig_k4}}}
		\end{minipage}
		\hspace{3mm}
		\begin{minipage}[b]{86pt}
			\centering
			{\includegraphics[width=86pt]{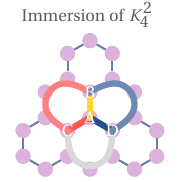}\subcaption{\label{fig_hex_k42}}}
		\end{minipage}
		\begin{minipage}[b]{86pt}
			\centering
			{\includegraphics[width=86pt]{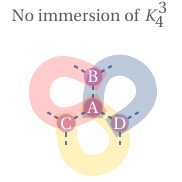}\subcaption{\label{fig_hex_k43}}}
		\end{minipage}    
		\hspace{3mm}
		\begin{minipage}[b]{86pt}
			\centering
			{\includegraphics[width=86pt]{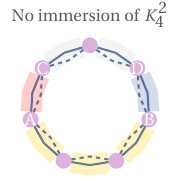}\subcaption{\label{fig_cycle_k42}}}
		\end{minipage}
		\begin{minipage}[b]{86pt}
			\centering
			{\includegraphics[width=86pt]{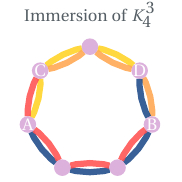}\subcaption{\label{fig_cycle_k43}}}
		\end{minipage}
		
		\caption{\textbf{$K_4^2$ versus~$K_4^3$}. \subref{fig_k4}~Both $K_4^2$ and $K_4^3$ are three-edge-connected, i.e.,~every pair of vertices (e.g.,~A and B) are connected by three edge-disjoint paths (red/pink, blue/navy, yellow/orange). But $K_4^2$ and $K_4^3$ cannot be immersed into each other.
			\subref{fig_hex_k42}~In an infinite honeycomb lattice, there exists an immersion of $K_4^2$, \subref{fig_hex_k43}~but not $K_4^3$. \subref{fig_cycle_k42}~Conversely, in a cycle graph made up of $\ge 4$ vertices and double-edges, 
			there is no immersion of $K_4^2$ \subref{fig_cycle_k43}~but $K_4^3$.
			\hfill\hfill
		}
		\vspace{-1mm}
	\end{figure*}
	
	\emph{Examples.---}Let the resulted multipartite QN after entanglement routing be a complete $r$-uniform hypergraph, $H=K_{n}^r$, 
	where $n\equiv\left|V(K_{n}^r)\right|$, and the edge set $E(K_{n}^r)$ is equal to the set of all size-$r$ subsets in $V(K_{n}^r)$.  In other words, $K_{n}^r$ includes all possible hyperedges of size $r$. 
	We present the first few necessary and sufficient conditions for $H=K_n^r$ to be immersed in $G$ in Table~\ref{table_k}. We observe that the conditions are straightforward for most small $n$ and $r$, except for $K_4^2$ and $K_4^3$, which represent the first nontrivial examples. 
	Indeed, while both $K_4^2$ and $K_4^3$ are {three-edge-connected} (SM, Section~1), 
	it can be shown that neither one can be immersed in the other [Fig.~\ref{fig_k4}]. This is an explicit example showing that immersion is a partial order, not a total order. 
	
	We discuss the criteria for immersion of $K_4^2$ and $K_4^3$ further in the SM. 
	These criteria lead to more compelling examples, involving infinite network topologies:

	\begin{theorem}
		There is an immersion of $K_{4}^2$ but not $K_{4}^3$ in an infinite honeycomb lattice ($6^3$).
		\label{the:honeycomb-1}
	\end{theorem}
	\begin{proof}
		The proof of the existence of an immersion $\alpha$ of $H=K_{4}^2$ in an infinite honeycomb lattice $G$ is straightforward. According to Table~\ref{table_k}, since $G$ is not series-parallel, $\alpha$ must exist. An illustration of such an immersion is presented in Fig.~\ref{fig_hex_k42}.
		
		For $H=K_{4}^3$, assume there exists an immersion $\alpha$ of $H$ in $G$, with the four vertices in $H$ mapped to $A$, $B$, $C$, and $D$ in $G$. The presence of a hyperedge of size $3$ implies connectivity among the three vertices $\{A,B,C\}$ in $G$. Essentially, there are three possible minimal configurations: $A$-$B$-$C$, $B$-$A$-$C$, or $A$-$C$-$B$, where, for instance, $A$-$B$-$C$ indicates a path from $A$ to $B$ and then from $B$ to $C$. Similarly, for vertex sets $\{A,B,D\}$ and $\{A,C,D\}$, they must be similarly connected through paths in $G$. All the paths must be edge-disjoint. However, in a honeycomb lattice, each vertex is contained in at most three edge-disjoint paths. Consequently, there are only two possible path configurations for three hyperedges, depicted in Fig.~\ref{fig_hex_k43} (another one being its mirror symmetry). Subsequently, vertices $B$, $C$, and $D$ have exhausted all available edge-disjoint paths originating from them. Therefore, there is no space in $G$ for the immersion of a fourth hyperedge connecting $\{B,C,D\}$, contradicting $\alpha$'s existence.
	\end{proof}

	\begin{theorem}
		There is an immersion of $K_{4}^3$ but not $K_{4}^2$ in a cycle graph consisting of no less than four vertices and only double-edges.
	\end{theorem}
	\begin{proof}
		The proof of the existence of an immersion $\alpha$ of $H=K_4^3$ in a cycle graph consisting of double edges and no less than $4$ vertices ($G$) can be shown by construction. See Fig.~\ref{fig_cycle_k43} for an example in $G$ of size $7$, where the four vertices in $H$ are mapped to $A,B,C,D$ in $G$. The generalization to $G$ of an arbitrary size $n$ is straightforward. 
		
		For $H=K_4^2$, assume by contradiction that such immersion exists, denoted by $\alpha$. Assume w.l.o.g.~the four vertices $A',B',C',D'$ in $H$ are mapped to $A,B,C,D$, respectively, that are consecutive anti-clockwise in $G$. The presence of edges $A'B', A'C', A'D'$ implies three edge-disjoint paths between $A,B$, $A,C$, and $A,D$, respectively. Since node $C$ can only reach node $A$ through nodes $B$ or $D$, assume w.l.o.g.~that $B\in \alpha(A'C')$. Hence, there are two edges that are incident on node $B$ in $\alpha(A'C')$, and there is also one edge that is incident on node $B$ in $\alpha(A'B')$. Then, there is only one edge in $G$ that is incident on node $B$ but is not in $\alpha(A'B')$ or $\alpha(A'C')$. However, $\alpha(B'C')$ and $\alpha(B'D')$ imply that there are two edges that are incident on node $B$ but are not in $\alpha(A'B')$ or $\alpha(A'C')$, which leads to a contradiction.    
	\end{proof}
	
	These examples have profound implications for QN designs. For instance, we have shown that the 2D lattice in Fig.~\ref{fig_hex_k42} is impossible for simultaneous secret sharing among every three parties in any four-party group through entanglement routing, unless additional quantum channels~\cite{q-resour_kl21} or memories~\cite{q-mem_agr15,*q-netw-mem_mpnk24} are introduced to the QN to create more links. Similarly, the quasi-1D structure in  Fig.~\ref{fig_cycle_k43} also has inherent topological limitations. 
	To conclude, 
	the hypergraph immersion problem we introduce captures the topological landscape of entanglement routing in QNs, which could allow for analyzing very large scales (due to its polynomial complexity~\cite{Note1}).
	
	\emph{Discussion.---}In classical communication, information transmission is fundamentally \emph{directional}~\cite{inf-theor}, reflecting the causality between the sender and receiver. Hence, a classical broadcast can only disseminate information from one-to-many, not many-to-one. In contrast, quantum entanglement lacks a causal structure to transmit classical information, as suggested by the no-communication theorem~\cite{douglas2001noncommutative}. This absence of causality underscores the unique interpretation of multipartite entangled states as undirected multiedges, a pure quantum analog. 
	Also, it is worth noting that a variant of immersion involves \emph{labeled} vertices in both hypergraphs $G$ and $H$, corresponding to the establishment of multipartite entanglement between specified vertices. In this labeled scenario, identifying an immersion of $H$ in $G$ is usually simpler than in the unlabeled case. We leave this direction for future exploration.

	
	\bibliography{refs}

	\newpage
	\clearpage
	\appendix
	
	\section*{Supplemental Materials}
	
	\renewcommand{\appendixname}{Section}
	\renewcommand{\thesection}{\Roman{section}}
	\renewcommand{\theequation}{\Roman{section}.\arabic{equation}}
	\renewcommand{\thesubsection}{\Alph{subsection}}

	\renewcommand
	\thefigure{S\arabic{figure}}
	
	\setcounter{figure}{0}
	\setcounter{equation}{0}

	\section{Preliminaries}
	A \emph{graph} $G$ consists of vertices and edges. We denote by $V(G)$, $E(G)$ for the set of vertices and edges, respectively. We consider undirected graphs in this paper, where multiple edges may be present among the same set of vertices, but loops, i.e., an edge starting and ending at the same node, are ignored.
	In an ordinary graph, an edge can only connect two vertices, thus $\forall e\in E(G)$, $e=uv$ with $u,v\in V(G)$. 
	
	End-to-end connections between vertices are important in a graph. Specifically, a \textit{cut point} in a connected graph $G$ is a vertex whose removal disconnects $G$. Two vertices in $G$ are said to be \textit{biconnected} if they cannot be disconnected by the removal of any cut point. A {\textit{biconnected component}} of $G$ is the subgraph induced by a maximal set of pairwise biconnected vertices. 
	
	Two vertices are said to be \textit{three-edge-connected} if there are at least three edge-disjoint paths between them. A \textit{three-edge-connected component} of $G=(V,E)$ is a graph $G'=(V',E')$ where $V'\subseteq V$ is a maximal set of vertices that are pairwise three-edge-connected in $G$, and $E'$ contains all edges induced by $V'$ plus a (possibly empty) set of \textit{virtual edges} defined as follows: for $\{u,v\}\subseteq V'$, a virtual edge $\{u,v\}$ is added to $E'$ whenever there exist a pair of distinct nodes $\{x,y\}$ that are not in $V'$, such that removing edges $\{u,x\}$ and $\{v,y\}$ disconnects the graph. 
	{The three-edge-connected component better represents the connectivity of the nodes, projected onto a smaller graph.}
	Note that, due to the possible presence of virtual edges, a three-edge-connected component will not necessarily be a subgraph. 
	
	Another important definition is based on the series and parallel operations. Specifically, an edge is said to be in \textit{series-parallel connection} if the joint resistor (of the whole graph) through this edge can be evaluated by Ohm's two rules (resistors either in series or in parallel)~\cite{Duffin_seriesparallel_1965}. A graph is \textit{series-parallel} if every edge is in series-parallel connection.
	
	We consider the important notion of \emph{immersion} in this paper. 
	Specifically, for ordinary graphs, an immersion of $H$ in $G$ is a function $\alpha$ with domain $V(H)\cup E(H)$ such that~\cite{RobertsonSeymour_minor_2010}: 
	\begin{enumerate}
		\item $\alpha(v) \in V(G)$ for all $v\in V(H)$, and $\alpha(u)\ne\alpha(v)$ for all distinct $u,v\in V(H)$;
		\item for each edge $e\in E(H)$, if $e$ has distinct ends $u,v$ then $\alpha(e)$ is a path of $G$ with ends $\alpha(u), \alpha(v)$; 
		\item for all distinct $e_1,e_2\in E(H)$, $E(\alpha(e_1)\cap \alpha(e_2))=\emptyset$;
	\end{enumerate} 
	Equivalently, an immersion of $H$ in $G$ exists if and only if $H$ can be obtained from (a subgraph of) $G$ by \textit{lifting}, an operation on adjacent edges: given three vertices $v$, $u$, and $w$, where $\{v,u\}$ and $\{u,w\}$ are edges in the graph, the lifting of $v,u,w$, or equivalently of $\{v,u\}, \{u,w\}$ is the operation that deletes the two edges $\{v,u\}$ and $\{u,w\}$ and adds the edge $\{v,w\}$. 
	
	The notion of \emph{embedding} (a.k.a.~topological minor) is closely related to immersion. Specifically, an embedding of $H$ in $G$ is a function $\beta$ that is defined the same as an immersion $\alpha$, except that in rule (iii), $E(\beta(e_1)\cap \beta(e_2))=\emptyset$ is replaced by $V(\beta(e_1)\cap \beta(e_2))=\beta(V(e_1 \cap e_2))$ where $V(e_1 \cap e_2)$ denotes the common endpoints of $e_1$ and $e_2$. In other words, $\beta$ is an injective mapping from edges in $H$ to vertex-disjoint paths, instead of edge-disjoint paths as in immersion. 
	As a special example, it is well known that an embedding of a complete graph of size $4$ exists if and on if the graph is not series-parallel \cite{Duffin_seriesparallel_1965}. 
	We note that an embedding is an immersion, but \emph{not} vice versa.
	
	A \emph{hypergraph} relaxes the constraints on the number of vertices in each edge, where a hyperedge $e$ can connect an arbitrary number of vertices, and we denote the size of a hyperedge by the number of vertices in this edge. An {$r$-uniform hypergraph} is then a hypergraph where each hyperedge has size $r$, and a {complete $r$-uniform hypergraph} is the one with all possible edges of size $r$, denoted by $K_n^r$ where $n$ encodes the size of the hypergraph, i.e.,~$\left|V(K_{n}^r)\right| = n$. 
	
	The end-to-end connection in hypergraphs can be considered through the generalization of a path in an ordinary graph to hypergraphs, the \emph{Berge path}~\cite{berge1973hyper}. A \emph{Berge path} of length $t$ is an alternating sequence of distinct $t+1$ vertices and distinct $t$ hyperedges of the hypergraph $G$, $v_1, e_1, v_2, e_2, v_3,\cdots, e_t, v_{t+1}$, such that $v_i, v_{i+1}\in e_i$, for $i=1,\cdots,t$. 
	A \emph{Berge cycle} is the same as a Berge path except that the last vertex $v_{t+1}$ is replaced by $v_0$, making the sequence cyclic.
	A \emph{connected subgraph} of a hypergraph is then a set of vertices which are pairwise connected by some Berge path(s). The size of the subgraph is the size of its vertex set.
	
	\section{Criteria on Immersion of $K_4^2$ or $K_4^3$}

	This section details how to determine whether $K_4^2$ or $K_4^3$ can be immersed in a hypergraph $G$. 
	Our main strategy is to transform the hypergraph immersion problem for $G$ into a set of (simpler) immersion problems on a new set of ordinary graphs, denoted $\{G^{\text{dw}}_i\}_i$ (Fig.~\ref{fig_gdw}). Each $G^{\text{dw}}_i$ is uniquely derived from $G$, containing only size-$2$ edges. This reduction simplifies the problem, as the immersion of the edges in $G^{\text{dw}}_i$ can then be treated 
	using the lifting operations, opening up the possibility of leveraging existing techniques and theorems.
	
	While the following criteria---as we will present below---are intuitive, they do not readily generalize to hypergraphs beyond $K_4^2$ and $K_4^3$. Consequently, \emph{a general reduction of hypergraph immersion to ordinary graph immersion is not straightforward}.
	Moreover,
	the associated testing procedures may not be as efficient and can exhibit exponential time complexity w.r.t.~the size of $G$. A proof demonstrating the existence of a more efficient, polynomial-time algorithm for general hypergraph immersion can be found in \xiangyi{Ref.~\footnote{In preparation.}.}

	\begin{figure}[t!]
		\centering
		\includegraphics[width=243pt]{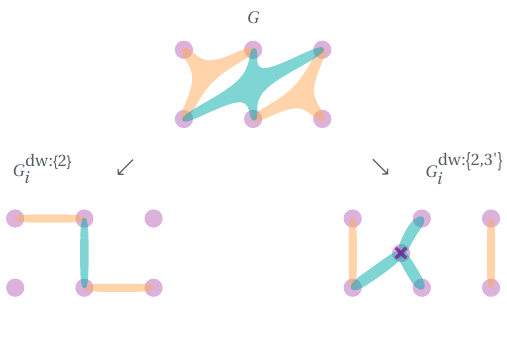}
		\caption{\textbf{Conversion of a hypergraph $G$ into sets of ordinary graphs $\{G^{\text{dw}}_i\}_i$.} Two example graphs are shown, one from the set $\{G^{\text{dw:}\{2\}}_i\}_i$ and one from $\{G^{\text{dw:}\{2,3'\}}_i\}_i$.
			\hfill\hfill}\label{fig_gdw}
	\end{figure}
	
	\subsection{$K_4^2$}
	To determine if $H=K_4^2$ can be immersed in $G$, we consider the set $\{G^{\text{dw:}\{2\}}_i\}_i$ of all possible ordinary graphs derivable from $G$ solely through dewetting operations (Fig.~\ref{fig_gdw}). Each graph $G^{\text{dw:}\{2\}}_i$ retains the original number of hyperedges (i.e.,~$\left|E(G^{\text{dw:}\{2\}}_i)\right|=\left|E(G)\right|$), but after dewetting, each hyperedge $e \in E(G)$ loses connection to all but only two incident nodes, becoming an ordinary edge. The following theorem provides a criterion for the immersion of $K_4^2$:
	
	\begin{theorem}
		\label{theorem_k42}
		The complete ($2$-uniform) graph $H=K_4^2$ can be immersed in a hypergraph $G$ if and only if there exists a subgraph of $G^{\text{dw:}\{2\}}_i$ for some $i$, such that $K_4^2$ can be constructed by lifting operations applied to the edges of the subgraph of $G^{\text{dw:}\{2\}}_i$.
	\end{theorem}
	\begin{proof}
		The sufficient condition (``if'') is straightforward. Since the lifting operation is nothing but a combination of the coalescence and dewetting operations for ordinary graphs, it implies that $H=K_4^2$ is derivable by coalescence and dewetting operations on a subgraph of $G^{\text{dw:}\{2\}}_i$, which is derivable by dewetting operations on $G$. This, by definition, constitutes an immersion of $K_4^2$ in $G$.
		
		The necessary condition (``only if'') goes as follows. Given an immersion $\alpha$ of $H$ in $G$, for each edge $e=\{u,v\}\in H= K_4^2$, 
		the corresponding connected subgraph $\alpha(e)$ in $G$ must contain a Berge path that connects $u$ and $v$, which can be derived from $\alpha(e)$ through dewetting. This Berge path, composed entirely of size-2 edges, can be lifted to form the corresponding edge $e \in K_4^2$. Since for all distinct $e_1,e_2\in E(H)$, $\alpha(e_1)$ and $\alpha(e_2)$ are edge-disjoint, the corresponding Berge paths are also edge-disjoint. Therefore, the union of the Berge paths of $\alpha(e)$ for every $e\in E(H)$ 
		must be a subgraph of some member of the set
		$\{G^{\text{dw:}\{2\}}_i\}_i$.
	\end{proof}
	
	Consequently, the hypergraph immersion problem for $G$ is reduced to a set of ordinary graph immersion problems, one for each $G^{\text{dw:}\{2\}}_i$. The immersion of $H=K_4^2$ in ordinary graphs is a well-studied problem~\cite{Booth_K4_1999}. 
	Specifically, if $G^{\text{dw:}\{2\}}_i$ is non-series-parallel (see Ref.~\cite{ValdesEtal_2TSP_1982} for an efficient algorithm), then an embedding of $K_4^2$ exists, and thus an immersion of $K_4^2$ exists; 
	Else, for every three-edge-connected components of each $G^{\text{dw:}\{2\}}_i$, run Algorithm~\ref{alg:test-G} to check whether any of them falls into the special cases where an immersion of $K_4^2$ exists (see Ref.~\cite{Booth_K4_1999} for more details on constructing three-edge-connected components). 
	
	\begin{algorithm}[H]
		\caption{test($X$).}
		\label{alg:test-G}
		\begin{algorithmic}[1] 
			\State{Input: a three-edge-connected series-parallel graph $X$.} 
			\State{Output: YES, if $X$ contains an immersed $K_4$, NO otherwise.} 
			\For{each vertex $v$ in $X$ with exactly one neighbour}
			\State delete all but three copies of edges incident on $v$
			\EndFor
			\If{any cut point in $X$ has degree $7$ or more}
			\State output YES and halt
			\EndIf
			\For{each biconnected component $B$ with four or more vertices}
			\State prune $B$ 
			\If{there is a vertex in $B$ with degree $5$ or more}
			\State output YES and halt
			\EndIf
			\EndFor
			\State output NO and halt
		\end{algorithmic}
	\end{algorithm}
	
	In particular, 
	the \emph{pruning} operation in Algorithm~\ref{alg:test-G} is defined as follows: in a three-edge-connected series-parallel multigraph, suppose $v$ is a vertex with exactly two neighbours $u$ and $w$, and suppose there is only one copy of edge $\{u,w\}$ (then at least two copies of $\{u,v\}$ by three-edge-connectivity); we say that $v$ is \emph{pruned} if the multiplicity of $\{u,v\}$ is set to $2$, and we say that a graph is pruned if each vertex fitting the profile of $v$ is pruned. 

	\begin{figure*}[t!]
		\centering
		\includegraphics[width=397.2pt]{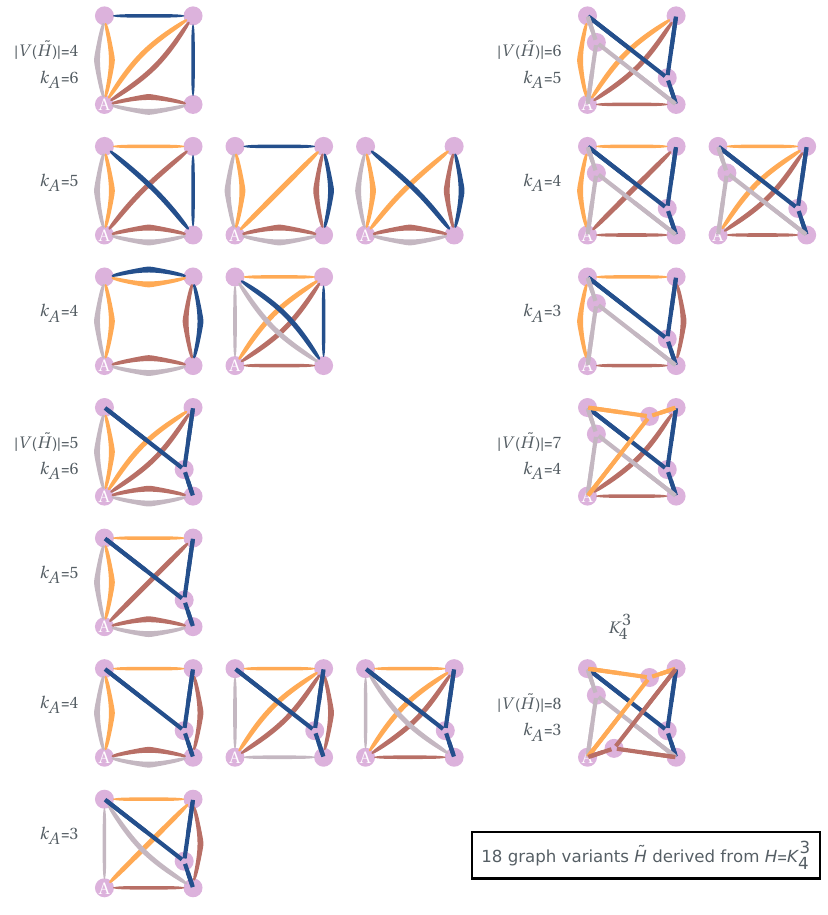}
		\caption{\textbf{Variants of $H=K_4^3$.} The set $\mathcal{F}(H)$ contains $18$ distinct graph topologies derived from $K_4^3$. Each size-$3$ hyperedge in $K_4^3$ is replaced by either two edges connecting the three incident vertices or a Y-shaped graph with an additional vertex.\hfill\hfill}\label{fig_variants_factor}
	\end{figure*}

	\subsection{$K_4^3$}
	

	Similarly, to determine if $H=K_4^3$ can be immersed in $G$, we consider the set  of all possible hypergraphs derivable from $G$ solely through dewetting operations, such that each hyperedge $e \in E(G)$ reduces to either a size-$2$ edge or a size-$3$ hyperedge. Then, we replace each size-$3$ hyperedge $e$ by its factor graph $e'$, i.e.,~a Y-shaped graph with an additional central vertex denoted by ``$\times$'' (Fig.~\ref{fig_gdw}). The set of all possible graphs derived this way is denoted by  $\{G^{\text{dw:}\{2,3'\}}_i\}_i$.
	
	
	
	In addition, we construct a set of ordinary graphs, $\mathcal{F}(H)$, derived from $H=K_4^3$ by replacing each size-$3$ hyperedge with either two ordinary edges connecting the three incident vertices or a Y-shaped graph (three edges connecting the vertices to an additional central vertex). Up to isomorphism, this yields $18$ distinct graph topologies, illustrated in Fig.~\ref{fig_variants_factor}. It is straightforward to see that for every $F\in \mathcal{F}(H)$, there exists an immersion of $H$ in $F$.
	
	The following theorem provides a criterion for the immersion of $K_4^3$:

	\begin{theorem}
		\label{theorem_k43}
		The complete $3$-uniform hypergraph $H=K_4^3$ can be immersed in a hypergraph $G$ if and only if there exists a subgraph of $G^{\text{dw:}\{2,3'\}}_i$ for some $i$, such that at least one of the $18$ graphs $F\in\mathcal{F}(H)$ (Fig.~\ref{fig_variants_factor}) can be derived from the subgraph of $G^{\text{dw:}\{2,3'\}}_i$ by lifting operations, with the restriction that no vertex in $H$ will be mapped to a vertex in $G^{\text{dw:}\{2,3'\}}_i$ that has a ``$\times$'' label.
	\end{theorem}
	\begin{proof}
		Similar to Theorem~\ref{theorem_k42},  the sufficient condition (``if'') is as follows. If one of the graphs $F \in \mathcal{F}(H)$ can be immersed in $G^{\text{dw:}\{2,3'\}}_i$, then by construction, $H$ can also be immersed in $G^{\text{dw:}\{2,3'\}}_i$. We must now verify that this immersion remains valid when extended to $G$. Since ``$\times$'' vertices are not actual vertices in $G$, two potential issues could invalidate the immersion: (1) a vertex of $H=K_4^3$ is mapped to a ``$\times$'' vertex, or (2) a lifting operation involving a ``$\times$'' vertex creates a topology not derivable from $G$ through coalescence and dewetting. The first issue is precluded by the theorem's conditions. The second is also impossible: at most three edge-disjoint paths can terminate at a ``$\times$'' vertex. Lifting joins two of these paths, while the third becomes isolated and irrelevant to $H$'s immersion (since ``$\times$'' vertices cannot map to vertices in $H$). This isolated path can thus be discarded. The joined two paths, however, are equivalent to dewetting the original size-$3$ hyperedge in $G$ to a size-$2$ edge, followed by lifting operations. Therefore, the immersion of $H$ in $G^{\text{dw:}\{2,3'\}}_i$ (with the ``$\times$'' vertex preclusion) implies sufficient connectivity for $H$ to be immersed in $G$.
		

		
		The necessary condition (``only if'') is as follows. 
		For each size-$3$ hyperedge $e=\{u,v,w\}\in K_4^3$, 
		the corresponding connected subgraph $\alpha(e)$ in $G$ contains three Berge paths (from $u$ to $v$, $v$ to $w$, and $w$ to $u$, respectively). There are only three possibilities for these paths: (1) Two Berge paths are edge-disjoint (e.g.,~$u$ to $v$ and $v$ to $w$ w.l.o.g.). (2) No two Berge paths are edge-disjoint, and an additional vertex $x$ exists in $\alpha(e)$ with edge-disjoint sub-paths connecting $x$ to $u$, $x$ to $v$, and $x$ to $w$, respectively. (3) No two Berge paths are edge-disjoint, and at least one hyperedge in $\alpha(e)$ is shared by all three Berge paths.
		
		In case~(1), the two Berge paths can be derived from $\alpha(e)$ by dewetting to size-$2$ edges. In case~(2), similarly, the three sub-paths can be derived by dewetting to all size-$2$ edges. In case~(3), we dewet $\alpha(e)$ to all size-$2$ edges, except for one of the shared hyperedges. This remaining hyperedge is then dewetted to a size-$3$ hyperedge connecting exactly to $u$, $v$, and $w$ via edge-disjoint sub-paths. Replacing this size-3 hyperedge with its factor graph makes this case equivalent to the second case, with the central vertex corresponding to $x$. Since in all three cases the Berge paths (or sub-paths involving $x$) are edge-disjoint, the union of them must be a subgraph of some member of the set $G^{\text{dw:}\{2,3'\}}_i$.
		
		Finally, the construction of $\mathcal{F}(H)$ enumerates all possible combinations of these cases (two-edge-series graph or Y-shaped graph) for each hyperedge, up to isomorphism. This completes our proof.
	\end{proof}
	
	Note that the preceding proof, particularly the sufficient condition, applies only when $H$ contains hyperedges of size at most $3$. Consider a size-$4$ hyperedge, $e = tuvw$, for example. Lifting the central vertex ``$\times$'' in its factor graph would create two edges, $tv$ and $uw$ w.l.o.g. However, two distinct edges cannot be derived from a single hyperedge $e$ through coalescence and dewetting, and this would invalidate the sufficient condition in Theorem~\ref{theorem_k43}. 
	This suggests that the connection between hypergraph and ordinary graph immersion is indeed more subtle, and one cannot simply reduce a hypergraph immersion problem to ordinary graph immersion for the general case.
	Therefore, alternative techniques are required.
	

	\section{More Examples on $K_4^2$ and $K_4^3$}
	In this section, we give more nontrivial examples when the hypergraph has an immersion of $K_4^2$ and/or $K_4^3$. 
	
	\begin{figure}[t!]
		\begin{minipage}[b]{121pt}
			\centering
			{\includegraphics[width=121pt]{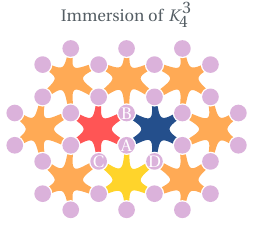}\subcaption{\label{fig_hyperhex_k43}}}
		\end{minipage}
		\begin{minipage}[b]{121pt}
			\centering
			{\includegraphics[width=121pt]{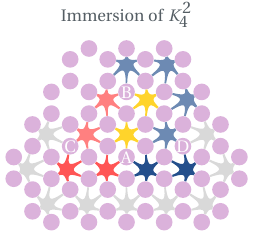}\subcaption{\label{fig_hyperhex_k42}}}
		\end{minipage}
		\centering 
		\caption{
			There exist both immersions of $K_4^2$ and $K_4^3$ in an infinite honeycomb hyper-lattice, but of different depth.
			\hfill\hfill
		}
		\label{fig_hyper}
	\end{figure}
	
	Now, if we consider the honeycomb lattice again, but instead of pairwise connections, vertices are connected by hyperedges of size $6$ (or hexagons), we can show that an immersion of $K_{4}^3$ will occur; see Fig.~\ref{fig_hyperhex_k43} for the construction. We denote such structure as infinite honeycomb hyper-lattice. Following a similar idea to the proof of Theorem~\ref{the:honeycomb-1} in the main text, we can show that an immersion of $K_4^2$ still exists; see Fig.~\ref{fig_hyperhex_k42} for the construction.
	It is clear that for every image of $K_4^2$ in an infinite honeycomb hyper-lattice, there is an immersion of $K_4^3$, but it is not necessarily true for the opposite; see Fig.~\ref{fig_hyper} for an example when three out of four images of the four vertices in $K_4^2$ are (directly) connected with the remaining one.  
	Specifically, if we restrict the honeycomb hyper-lattice to a smaller region where we cannot find four vertices such that (i) none of them are on the boundary, (ii) any two of them do not share a hexagon, and (iii) at most two of three hexagons that are incident on three of the four vertices are on the boundary, then we cannot find an immersion of $K_4^2$. 
	Hence, if we consider hypergraphs that are locally honeycomb hyper-lattice like (satisfying the above condition), and globally tree like, then an immersion of $K_4^2$ does not exist but $K_4^3$. 

	We finish this section with a slightly different type of graphs consisting of hyperedges of size $4$ that is necessarily shown in 3-dimensional space (or is not planar as before and in the main text). 
	\begin{figure}[t!]
		\centering
		\begin{tabular}{c}
			\includegraphics[width=.45\textwidth]{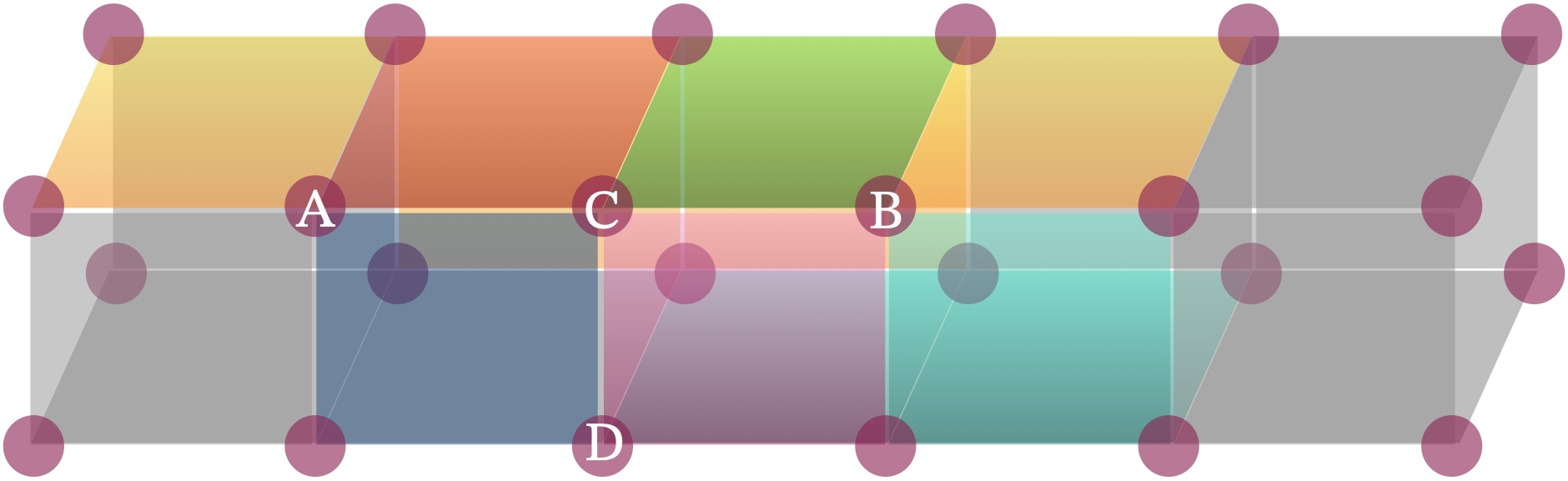} \\
			\vspace*{1em}\\
			\includegraphics[width=.45\textwidth]{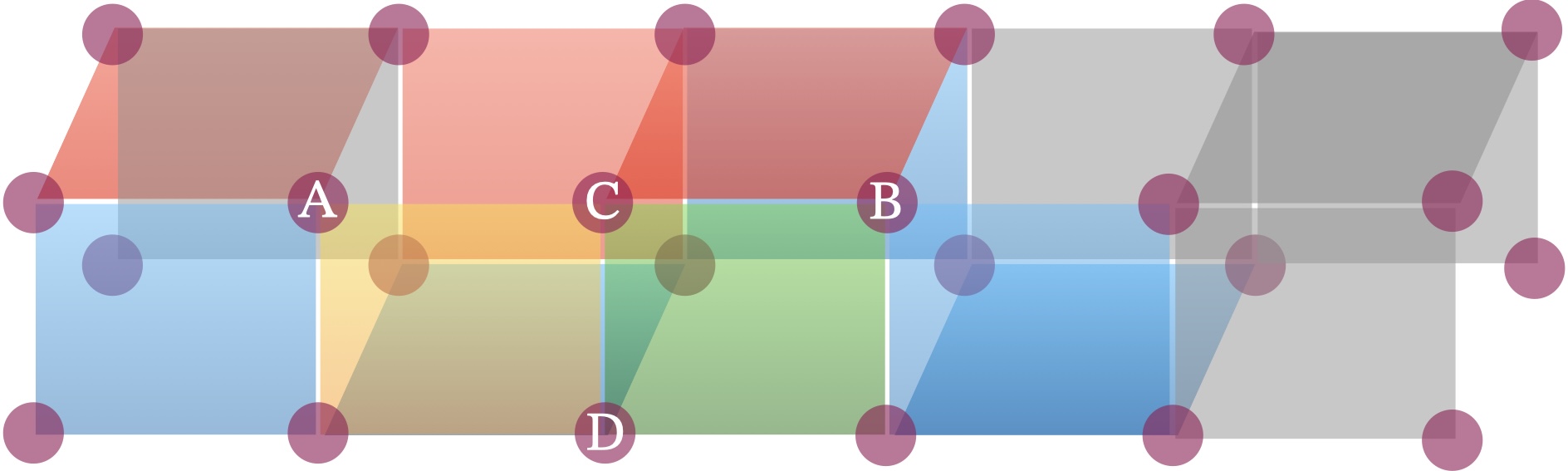}
		\end{tabular}
		\caption{\textbf{Cuboid full lattice versus cuboid alternating lattice}. Hypergraphs in the shape as a cuboid where there is no edges inside but on the four outer faces: (top) four faces are filled with hyperedges of size $4$, cuboid full lattice; (bottom) there can only be on hyperedge in the two opposite faces, cuboid alternating lattice.\hfill\hfill}
		\label{fig:K_43-sp}
	\end{figure}
	
	Specifically, we consider the hypergraph in the shape as an infinite cuboid where there is no edges inside but on the four infinite faces; see Fig.~\ref{fig:K_43-sp} for part of it. 
	We first consider the case when the four faces are filled with hyperedges (i.e.,~the top of Fig.~\ref{fig:K_43-sp}), and we denote such structure as infinite cuboid full lattice hereafter. 
	There is an immersion of $K_{4}^2$ in an infinite cuboid full lattice, since there exists an ordinary subgraph obtained with a series of dewetting operations on the hypergraph that is not series-parallel. Hence, there is an embedding of $K_4^2$, thus an immersion of $K_4^2$, in such hypergraph.
	
	However, we now consider the case when the hyperedges only occur in an alternating manner such that for the two opposite faces, there can only be one hyperedge in the two faces (i.e.,~the bottom of Fig.~\ref{fig:K_43-sp}), and we denote such structure as infinite cuboid alternating lattice hereafter. Then, we cannot find an immersion of $K_4^2$. 
		
	
	
	When it comes to $K_4^3$, it has an immersion in both hypergraphs; see the bottom of Fig.~\ref{fig:K_43-sp} for an example where vertices $A,B,C,D$ correspond to the four vertices in $H=K_4^3$. We also note that the infinite cuboid alternating lattice is a subgraph of the infinite cuboid full lattice. 

\end{document}